\documentclass[12pt]{article}
\pdfoutput=1
\usepackage[a4paper, margin = 34mm, top = 31mm, bottom = 38mm]{geometry}
\RequirePackage[hyphens]{url}

\usepackage{amsmath}
\usepackage{amssymb}
\usepackage{latexsym}

\RequirePackage[raiselinks=false,colorlinks=true,citecolor=blue,urlcolor=blue,linkcolor=blue,bookmarksopen=true]{hyperref}
\newcommand{\arxiv}[1]{\href{http://arxiv.org/pdf/#1}{arXiv:#1}}

\newcommand{\psy}{\psi}
\newcommand{\Lin}{\operatorname{Lin}}
\newcommand{\Ker}{\operatorname{Ker}}

\newcommand{\divides}{|}
\renewcommand{\int}{\operatorname{int}}

\newcommand{\phie}{\phi'}

\renewcommand{\epsilon}{\varepsilon}
\renewcommand{\Im}{\operatorname{Im}}
\newcommand{\CC}{\mathcal{C}}
\newcommand{\PP}{\mathcal{P}}
\newcommand{\fat}[1]{\mathbf{{#1}}}

\providecommand{\norm}[1]{\left| {#1}\right|}

\newcommand{\tensor}{\otimes}

\newcommand{\ground}{\Bbbk}

\newcommand{\N}{\mathbb{N}}
\newcommand{\C}{\mathbb{C}}
\renewcommand{\P}{\mathbb{P}}

\newcommand{\brok}[2]{{\textstyle{\frac{#1}{#2}}}}

\usepackage{amsthm}

\newtheorem{lemma}{Lemma}[section]
\newtheorem{prop}[lemma]{Proposition}

\newtheorem{theorem}[lemma]{Theorem}
\newtheorem{cor}[lemma]{Corollary}

\newtheorem{taller}[lemma]{$\!\!$}

\newenvironment{blanko}[1]%
{\begin{taller}{\normalfont\bfseries  #1}\normalfont}%
{\end{taller}}
\usepackage{graphicx}

\makeatletter

\renewcommand{\section}{\@startsection {section}{1}{\z@}%
{-3.5ex \@plus -1ex \@minus -.2ex}%
{2.3ex \@plus.2ex}%
{\normalfont\large\bfseries}}

\renewcommand*{\l@section}[2]{%
  \ifnum \c@tocdepth >\z@
    \addpenalty\@secpenalty
    \setlength\@tempdima{1.5em}%
    \begingroup
      \parindent \z@ \rightskip \@pnumwidth
      \parfillskip -\@pnumwidth
      \leavevmode 
      \advance\leftskip\@tempdima
      \hskip -\leftskip
      #1\nobreak\hfil \nobreak\hb@xt@\@pnumwidth{\hss #2}\par
    \endgroup
  \fi}

\makeatother

\makeatletter
\newcommand{\subjclass}[2][2010]{%
  \let\@oldtitle\@title%
  \gdef\@title{\@oldtitle\footnotetext{#1 \emph{Mathematics subject classification.} #2}}%
}
\newcommand{\keywords}[1]{%
  \let\@@oldtitle\@title%
  \gdef\@title{\@@oldtitle\footnotetext{\emph{Key words and phrases.} #1.}}%
}
\makeatother

\begin{document}

\title{From M\"obius inversion to renormalisation}
\author{Joachim Kock\\
\footnotesize Universitat Aut\`onoma de Barcelona}
\date{}

\maketitle

\begin{abstract}
  This paper traces a straight line from classical
  M\"obius inversion to Hopf-algebraic perturbative renormalisation.
  This line, which is logical but not entirely historical, consists of
  just a few main abstraction steps, and some intermediate steps
  dwelled upon for mathematical pleasure.  The paper is largely
  expository, but contains many new perspectives on well-known
  results.  For example, the equivalence between the
  Bogoliubov recursion and the Atkinson formula is exhibited as a direct
  generalisation of the equivalence between the Weisner--Rota
  recursion and the Hall--Leroux formula for M\"obius inversion.
%
\end{abstract}

\subjclass[2010]{
81T15; 
11A25; 
16T10; 
16T15; 
}

\keywords{Perturbative renomalisation, M\"obius inversion, arithmetic 
function, coalgebra, bialgebra, Rota--Baxter operator}

\section*{Introduction}

The flavour of renormalisation concerning the present contribution is
the BPHZ renormalisation of perturbative quantum field theories,
introduced by Bogoliubov, Parasiuk, Hepp and Zimmermann (1955-1969),
and more precisely its Hopf-algebraic interpretation discovered by
Kreimer~\cite{Kreimer:9707029} in 1998.  Subsequent work of Connes, 
Kreimer~\cite{Connes-Kreimer:9808042,Connes-Kreimer:9912092},
Ebrahimi-Fard, Guo, Manchon and 
others~\cite{EbrahimiFard-Guo-Manchon:0602004,EbrahimiFard-Guo-Kreimer}, distilled the construction
into a piece of abstract algebra, involving characters of a Hopf
algebra with values in a Rota-Baxter algebra.  It has important
connections
with disparate subjects in pure mathematics, such as 
multiple zeta values, numerical integration, and stochastic 
analysis. The construction itself can be viewed from various
perspectives, such as that of Birkhoff decomposition and the
Riemann--Hilbert
problem~\cite{Connes-Kreimer:9808042,Connes-Kreimer:9912092}, the
Baker--Campbell--Hausdorff formula and Lie
theory~\cite{EbrahimiFard-Guo-Manchon:0602004}, or the abstract
viewpoint of 
filtered non-commutative Rota--Baxter
algebras~\cite{EbrahimiFard-Guo-Kreimer}.  There are excellent surveys of
these developments, such as
Ebrahimi-Fard--Kreimer~\cite{EbrahimiFard-Kreimer:0510} (focusing on
physical motivation), Manchon~\cite{Manchon:0408} (generous with
mathematical preliminaries on coalgebras and Hopf algebras), and the longer
survey of Figueroa and
Gracia-Bond\'ia~\cite{Figueroa-GraciaBondia:0408145} (particularly relevant
in the present context for exploiting also the combinatorial viewpoint of
incidence algebras).

The aim of the present expository paper is to derive the construction
as a direct generalisation of classical M\"obius inversion: after the 
abstraction steps from the classical M\"obius function 
via incidence algebras to abstract M\"obius 
inversion, the remaining step is just to add a Rota--Baxter operator 
to the formulae.  This is
very close in spirit to Kreimer's original
contribution~\cite{Kreimer:9707029}, where the counter-term was staged
as a twisted antipode, but the explicit interpretation in terms of
M\"obius inversion seems not to have been made before, and in any 
case deserves to be more widely known.  The
perspective is attractive for its simplicity, and leads to clean and
elementary proofs (and slightly more general results---bialgebras
rather than Hopf algebras).  For ampler perspectives and deeper 
connections to various areas of mathematics, we refer to the 
bibliography and the pointers given along the way.

\medskip

Before starting from scratch with M\"obius inversion in classical
number theory (\S \ref{sec:clas}), it is appropriate to begin in \S
\ref{sec:BPHZ} by indicating more precisely where we are going, with a
brief introduction to BPHZ renormalisation from an abstract viewpoint.
After M\"obius inversion for arithmetic functions in \S \ref{sec:clas}, we
move to M\"obius inversion in incidence algebras in \S 
\ref{sec:incidence}; we deal with
both posets and M\"obius categories.  In \S \ref{sec:Mabs} we
establish the abstract M\"obius inversion principle, for general
filtered coalgebras with the property that the zeroth piece is spanned
by group-like elements.  This is inspired by recent work on M\"obius
inversion in homotopical contexts.  Finally in \S \ref{sec:RB} we add
a Rota--Baxter operator to the abstract M\"obius inversion formulae.
This yields directly the Bogoliubov recursion of renormalisation, 
and simultaneously the Atkinson formula.

\section{Hopf-algebraic BPHZ renormalisation}
\label{sec:BPHZ}

Perturbative quantum field theory is concerned with expanding
the scattering matrix
into a sum over graphs.  The Feynman rules assign to each
graph of the theory an amplitude.  Unfortunately, for many graphs
with loops (non-zero first Betti number), the corresponding amplitude
is given by a divergent integral.
Renormalisation is the task of extracting meaningful finite values
from these infinities.

In the (modern account of the) BPHZ approach, the first step consists 
in introducing a formal parameter, the 
{\em regularisation parameter} $\epsilon$, in such a way that
the amplitudes no longer take values directly in the complex numbers 
but rather in the ring of Laurent series 
$\C[\epsilon^{-1},\epsilon]]$.
The amplitudes are now well defined:
the divergencies are expressed by series with a pole at
$\epsilon=0$.
The next step is to subtract counter-terms for `divergent' graphs.
The {\em minimal subtraction scheme} aims simply to subtract the pole
part, but the naive attempt---just subtracting the pole part for a given 
graph---turns out to be
too brutal, destroying important physical features of the Feynman 
rules.
The problem can be localised to the fact that a divergent graph may
itself have divergent subgraphs, and these sub-divergencies should be
sorted out first, before attempting at determining the counter-term
for the graph as a whole.  In the end, the correct procedure, found by
Bogoliubov and Parasiuk~\cite{Bogoliubov-Parasiuk} and fine-tuned and
proved valid by Hepp~\cite{Hepp:1966},
is a rather intricate recursive 
over-counting/under-counting procedure, of a flavour not unfamiliar to
combinatorists.  The development culminated with Zimmermann~\cite{Zimmermann}
finding a closed formula
for the counter-term, 
the famous {\em forest formula}, instead of a recursion.%
\footnote{Important as it is, the forest formula is not dealt with in
the present exposition, as it is not clear how it relates to {\em
general} M\"obius inversion; but see
\cite{Figueroa-GraciaBondia:0408145} and
\cite{Menous-Patras:1511.07403} for important insight in this
direction for certain special classes of Hopf algebras.}
One crucial property is that the renormalised Feynman
rule remains a character, just like the unrenormalised Feynman rule,
expressing the fundamental principle that the amplitude of two
independent processes is the product of the processes.  The
renormalised Feynman rule assigns to every graph a power series
without pole part, and the desired finite amplitude can finally be
obtained by setting $\epsilon$ to $0$.  This procedure, called BPHZ
renormalisation, is described in many textbooks on quantum field
theory and renormalisation (e.g.~\cite{Collins}, \cite{Muta}).

Kreimer's seminal discovery~\cite{Kreimer:9707029} is that the combinatorics 
in this procedure is encoded in a Hopf algebra of graphs $H$.
As a vector space, $H$ is spanned by 
all 1PI graphs of the given quantum field theory. 
The multiplication in $H$ is given by taking disjoint union of 
graphs.
The comultiplication $\Delta: H \to 
H\tensor H$ is given on connected 1PI graphs 
$\Gamma$ by
$$
\Delta(\Gamma) = \sum_{\gamma\subset \Gamma}  \gamma \tensor 
\Gamma/\gamma  ,
$$
where the sum is over all (superficially divergent) 1PI subgraphs
$\gamma$ (possibly not connected), and the quotient graph $\Gamma/\gamma$ is obtained by
contracting each connected component of $\gamma$ to a vertex (the
residue of $\gamma$).  Altogether, $H$ is a Hopf algebra, graded by 
loop number.
The regularised Feynman rules are characters $\phi:H \to A$ with
values in $A=\C[\epsilon^{-1},\epsilon]]$.  The
{\em BPHZ counter-term} $\phi_-$ is given by the recursive formula (for $\deg \Gamma>0$)
$$
\phi_-(\Gamma) = - R\big[ \phi(\Gamma) + \!\!\underset{\gamma\neq \emptyset, \gamma\neq\Gamma}{\sum_{\gamma\subset \Gamma}}
\phi_-(\gamma)\,\phi(\Gamma/\gamma)\big] ,
$$
or more conceptually:
$$
\phi_- \ = \ e \;-\; R\big[ \phi_- * (\phi\!-\!e)\big]   ,
$$
where $*$ is convolution of linear maps $H \to A$, where
$e$ is the neutral element for convolution,
and $R:A\to A$ is the
idempotent 
linear operator that to a Laurent series assigns its pole part.
The recursion is well founded, thanks to the grading of $H$: 
the convolution refers to taking out subgraphs via $\Delta$, and
the arguments to $\phi_-$ in the convolution are graphs
with strictly fewer loops than the input to $\phi_-$ on the left-hand 
side of the equation, 
since $(\phi-e)$ vanishes on graphs without loops.

The renormalised Feynman rule is finally given in terms of 
convolution%
\footnote{That $\phi_+$ can be written as a convolution was 
realised by Connes and Kreimer~\cite{Connes-Kreimer:9912092} (thus
exhibiting the renormalisation procedure as an instance of the general 
mathematical construction called Birkhoff decomposition). Previously 
$\phi_+$
was computed via an auxiliary construction known as Bogoliubov's 
preparation map.}
as
$$
\phi_+ := \phi_- * \phi .
$$
It takes values in $\Ker R = \C[[\epsilon]]$, so that it makes sense 
finally to set $\epsilon=0$ to obtain a finite amplitude for each 
graph.  The crucial fact that $\phi_-$ and $\phi_+$ are again characters
turns out to be a consequence of a special property of the operator 
$R$, namely the equation
\begin{equation}\label{RBintro}
R(x\cdot y) + R(x) \!\cdot\! R(y)= R\big( R(x) \!\cdot\! y + x 
\!\cdot\! R(y) 
\big) ,
\end{equation}
which is to say that $R$ is a {\em Rota--Baxter 
operator}.\footnote{Kreimer himself did isolate conditions on $R$
ensuring that $\phi_-$ and $\phi_+$ are again characters, but it was
Brouder who observed that these conditions can be formulated as a 
single ``multiplicativity constraint'', namely
\eqref{RBintro} (see \cite{Kreimer:9901099}, footnote 4); 
Connes and Kreimer~\cite{Connes-Kreimer:9912092} referred to this 
multiplicatitivity constraint.
Ebrahimi-Fard then pointed out that this constraint is the 
Rota--Baxter equation (the first published mention being 
\cite{EbrahimiFard:0207043}), and started to import results and 
methods from this 
mathematical theory.}


\medskip

The abstraction of these discoveries is the purely algebraic result that
for any graded Hopf algebra $H$ and for any commutative algebra $A$ with
a Rota--Baxter operator $R$ like this, the same procedure works to
transform a character $\phi:H \to A$ into another character $\phi_-$ such
that the convolution $\phi_+ := \phi_- * \phi$ takes values in the kernel
of $R$ (the abstraction of the property of being pole free).  This is 
the result we will arrive at in Section~\ref{sec:RB}, from the 
standpoint of M\"obius inversion.

\medskip

It must be stressed that this neat little piece of algebra is only a
minor aspect of perturbative renormalisation, as it does
not account for the analytic (or number-theoretic) aspects of Feynman
amplitudes, e.g.~the computation of the individual integrals.  The
merit of the Hopf-algebraic approach is rather to separate out the
combinatorics from the analysis, and explain it in a conceptual way.
It is also worth remembering that assigning a renormalised amplitude
to every graph is not the end of the story, because there are infinitely
many graphs (their number even grows factorially in the number of loops),
and in general the sum of all these finite amplitudes will 
still be a divergent series in the coupling constant. New techniques
are being applied to tackle this problem, such as resurgence 
theory (see for example~\cite{Dunne-Unsal}).  The present contribution
deliberately ignores all these analytic aspects.

\section{The classical M\"obius function}

\label{sec:clas}

\begin{blanko}{Arithmetic functions and Dirichlet series.}
  Write
$$
\N^\times =\{1,2,3,\ldots\}
$$
for the set of positive natural numbers.
An {\em arithmetic function} is just a function
$$
f: \N^\times \to \C
$$
(meant to encode some arithmetic feature of each number $n$).
To each arithmetic function $f$ one associates a {\em Dirichlet series}
$$
F(s) = \sum_{n\geq 1} \frac{f(n)}{n^s} ,
$$
thought of as a function defined on some open set of the complex plane.
The study of arithmetic functions 
in terms of their associated Dirichlet series is a central topic in 
analytic number theory~\cite{Apostol}. 
\end{blanko}

\begin{blanko}{The zeta function.}
  A fundamental example is the 
{\em zeta function}
\begin{eqnarray*}
  \zeta: \N^\times & \longrightarrow & \C  \\
  n & \longmapsto & 1.
\end{eqnarray*}
The associated Dirichlet series is the {\em Riemann zeta function}
$$
\zeta(s) = \sum_{n\geq 1} \frac1{n^s} .
$$
\end{blanko}

\begin{blanko}{Classical M\"obius inversion.\footnote{This is due to 
  M\"obius~\cite{Moebius:1832}, see Hardy and 
  Wright~\cite{Hardy-Wright}, Thm.~266.}}
  The classical M\"obius inversion principle says that
  \begin{eqnarray*}
    \text{if} \qquad f(n) & = & \sum_{d\divides n} g(d)\\
    \text{then} \quad g(n) & = & \sum_{d\divides n} f(d) \mu(n/d) ,
  \end{eqnarray*}
  where $\mu$ is the {\em M\"obius function}\footnote{According to Hardy and Wright~\cite{Hardy-Wright} (notes to Ch.~XVI), 
  the M\"obius function occurs implicitly in the work of Euler as 
  early as 1748.}
  \begin{equation}\label{muclas}
  \mu(n) = \begin{cases} 0 & \text{if \(n\) contains a square factor}\\
  (-1)^r & \text{if \(n\) is the product of $r$ distinct primes.}\end{cases}
  \end{equation}
\end{blanko}

\begin{blanko}{Example: Euler's totient function.}
  {\em Euler's totient function} is by definition
  $$
  \varphi(n) := \# \{ 1\leq k \leq n \mid (k,n)=1\} .
  $$
  It is not difficult to see that we have the relation
  $$
  n = \sum_{d\divides n} \varphi(d),
  $$
  so by M\"obius inversion we get a formula for $\varphi$:
  $$
  \varphi(n) = \sum_{d\divides n} d\; \mu(n/d).
  $$
\end{blanko}

\begin{blanko}{Dirichlet convolution.}
  A conceptual account of the M\"obius inversion principle is given 
  in terms of {\em Dirichlet convolution} for
  arithmetic functions:
  $$
  (f * g)(n) = \sum_{i\cdot j = n} f(i) g(j) ,
  $$
  which corresponds precisely to (pointwise) product of Dirichlet series.
  The neutral element for this convolution product is the arithmetic function
  $$
  \epsilon(n) = \begin{cases}
    1 & \text{if }n=1 \\
    0 & \text{else}.
  \end{cases}
  $$
  
  Now the M\"obius inversion principle reads more conceptually
  $$
  f = g * \zeta \qquad \Rightarrow  \qquad g = f * \mu ,
  $$
  and the content is this:
\end{blanko}
  
\begin{prop}
  The M\"obius function is the convolution inverse of the zeta function.
\end{prop}

\begin{blanko}{Example (continued).}
  Let $\iota$ denote the arithmetic
  function $\iota(n)=n$.  Its associated Dirichlet series is
  $$
  \sum_{n\geq 1} \frac n {n^s} = \zeta(s-1) .
  $$
  Restating the M\"obius inversion formula for Euler's totient $\varphi$ in terms of Dirichlet 
  convolution yields
  $$
  \iota = \varphi * \zeta \qquad \Rightarrow \qquad \varphi = \iota * \mu ,
  $$
  so that the Dirichlet series associated to $\varphi$ is
  $$
  \frac{\zeta(s-1)}{\zeta(s)} .
  $$
\end{blanko}

\section{Incidence algebras}

\label{sec:incidence}

In the 1930s, M\"obius inversion was applied in group theory by
Weisner~\cite{Weisner:1935} and independently by
Hall~\cite{Hall:1936}.\footnote{Hall defined and computed {\em 
Eulerian functions} of groups using M\"obius inversion in subgroup 
lattices.
For cyclic groups, this recovers Euler's totient function.} 
Both were motivated by the lattice of
subgroups of a finite group, but found it worth developing the theory
more generally; Weisner for complete lattices, Hall for finite posets.

In the 1960s,
Rota~\cite{Rota:Moebius}
systematised the theory extensively, in the setting of locally finite 
posets, and made M\"obius inversion a 
central tool in enumerative combinatorics.  The setting of posets is now
widely
considered the natural context for M\"obius inversion (see for 
example 
Stanley's book~\cite{Stanley:volI}).  Cartier and
Foata~\cite{Cartier-Foata} developed the theory for monoids with the
finite-decomposition property, and Leroux~\cite{Leroux:1976} unified
these contexts in the general notion of M\"obius category, reviewed below.
More recently, Lawvere and Menni~\cite{Lawvere-Menni} and 
G\'alvez, Kock, and
Tonks~\cite{Galvez-Kock-Tonks:1512.07573,Galvez-Kock-Tonks:1512.07577,
Galvez-Kock-Tonks:1512.07580} took Leroux's ideas further into category 
theory and homotopy theory.\footnote{Lawvere and Menni~\cite{Lawvere-Menni} gave an
`objective' version of Leroux's theory: this means working with the
combinatorial objects themselves instead of the vector spaces they
span.  The classical theory is obtained by taking cardinality.  One
advantage of this approach---beyond making all proofs natively
bijective---is that one can eliminate finiteness conditions, if just 
one refrains from taking cardinality: the
constructions work the same with infinite sets, and at this level, M\"obius
inversion works for {\em any} category, not just M\"obius categories.
More recently, G\'alvez, Kock and
Tonks~\cite{Galvez-Kock-Tonks:1512.07573,Galvez-Kock-Tonks:1512.07577,
Galvez-Kock-Tonks:1512.07580} discovered that simplicial objects more
general than categories admit incidence algebras and M\"obius
inversion, and passed to the homotopical context of simplicial
$\infty$-groupoids.  Where categories express the general ability to 
compose, their notion of {\em decomposition space} expresses the general ability to 
decompose, in a appropriate manner so as to induce a coassociative 
incidence coalgebra, and an attendant M\"obius inversion principle.
There are plenty of examples in combinatorics of
coalgebras and bialgebras which are the incidence coalgebra of a
decomposition space but not of a category or a poset.  An example
relevant to the present context is the Connes--Kreimer Hopf algebra of
rooted trees~\cite{Connes-Kreimer:9808042}, which is the incidence bialgebra of a decomposition 
space but not directly of a category~\cite{Galvez-Kock-Tonks:1512.07573}.}
These abstract developments were crucial for distilling out the 
perspectives of the present contribution.

\bigskip

We briefly recall the notions of incidence algebras and M\"obius 
inversion for posets and M\"obius categories.  All proofs will be 
deferred to the abstract setting of Section~\ref{sec:Mabs}.
Throughout, $\ground$ denotes a ground field, `linear' means 
$\ground$-linear, and $\tensor$ is short for $\tensor_\ground$.

\begin{blanko}{The incidence (co)algebra of a locally finite posets.}
  A poset $(P,\leq)$ is called  {\em locally finite} if all its intervals
  $[x,y] := \{ z\in P : x\leq z \leq y \}$ are finite. 
  The free vector space $C_P$
  on the set of intervals becomes a coalgebra $(C_P,\Delta,\epsilon)$ with
  comultiplication $\Delta: C_P \to C_P \tensor C_P$ defined by
  $$
  \Delta([x,y]) := \sum_{z\in [x,y]} [x,z] \tensor [z,y]
  $$
  and counit $\epsilon:C_P\to \ground$ defined as
  $$
  \epsilon([x,y]) := \begin{cases}
    1 & \quad \text{if } x=y \\
    0 & \quad \text{else. }
  \end{cases}
  $$

  The {\em incidence algebra} of $P$
  is the convolution algebra of $C_P$ (with values in the 
  ground field).  The multiplication is thus given by
  $$
  (\alpha * \beta) ([x,y]) = \sum_{z\in[x,y]} 
  \alpha([x,z])\beta([z,y]) ,
  $$
  and the unit is $\epsilon$.
\end{blanko}

\begin{blanko}{The zeta function.}
  The {\em zeta function} is defined as
  \begin{eqnarray*}
    \zeta: C_P & \longrightarrow & \ground  \\
    {}[x,y] & \longmapsto & 1 .
  \end{eqnarray*}
  (Note that this function is constant on the set of intervals, but of course not
  constant on the vector space spanned by the intervals.)
\end{blanko}

\begin{blanko}{Theorem (Rota~\cite{Rota:Moebius}\footnote{The result 
  was essentially proved already by Weisner~\cite{Weisner:1935} (but only for 
  complete lattices) and by Hall~\cite{Hall:1936} (but only for finite 
  posets).}).}\label{muRota}
  {\em 
  For any locally finite poset, 
  the zeta function is convolution invertible; its 
  inverse, called the {\em M\"obius function}
  $\mu:=\zeta^{-1}$, is given by the recursive formula
  $$
  \mu([x,y]) = \underset{\underset{z\neq y}{z\in [x,y]} 
  \phantom{xxxxxxxxxxxxxxxxxxxxxx}}{\begin{cases}
    1 & \qquad\text{if } x=y \\[6pt]
    \displaystyle{-
	\ \sum \ \ \mu([x,z])}
	& \qquad\text{if } x<y . 
  \end{cases}}
  $$}%
This is a recursive definition by length of intervals, well founded 
  because of the condition $z\neq y$.
\end{blanko}

\begin{cor}
  We have
    $$
  f = g * \zeta \qquad \Rightarrow  \qquad g = f * \mu .
  $$
In other words,
  \begin{eqnarray*}
    \text{if} \qquad f([x,y]) & = & \sum_{z\in[x,y]} g([x,z])\\
    \text{then} \quad g([x,y]) & = & \sum_{z\in[x,y]} f([x,z]) 
	\mu([z,y]) .
  \end{eqnarray*}
\end{cor}

In fact, Rota proved more:
\begin{blanko}{Theorem (Rota~\cite{Rota:Moebius}).}\label{psiRota}
  {\em $\phi :C_P\to \ground$ is convolution invertible provided 
  $\phi([x,x])=1$ for 
  all $x\in P$; the convolution inverse $\psi$ is determined by the 
  recursive formula
  $$
  \psi([x,y]) = \underset{\underset{z\neq y}{z\in [x,y]} 
  \phantom{xxxxxxxxxxxxxxxxxxxxxxxxxxxxxx|}}{\begin{cases}
    1 & \qquad\text{if } x=y \\[6pt]
    \displaystyle{-
	\ \sum \ \ \psi([x,z])} \,\phi([z,y])
	& \qquad\text{if } x<y . 
  \end{cases}}
  $$
  }
  
The recursion can be written more compactly as
$$\boxed{
\psi \ = \ \epsilon \;-\; \psi * ( \phi \!-\! \epsilon ) }
$$
Indeed, subtracting $\epsilon$ from $\phi$ inside the sum expresses 
the
fact that we don't want the last summand ($z=y$), and adding the term 
$\epsilon$ outside the sum expresses the first case ($x=y$).
\end{blanko}

\begin{blanko}{M\"obius categories (Leroux).}
  Leroux~\cite{Leroux:1976} introduced the common generalisation of
  locally finite posets and Cartier--Foata monoids: {\em M\"obius
  categories}.  Recall that posets and monoids are special cases of
  categories:  a poset can be considered as a category whose objects
  are the elements of the poset, and in which there is an arrow from
  $x$ to $y$ if and only if $x\leq y$ in the poset.  This means that
  arrows now play the role of intervals, and splitting intervals
  becomes factorisation of arrows.  A monoid can be considered as a
  category with only one object, the arrows being then the monoid
  elements, composed by the monoid multiplication.  The finiteness
  conditions for posets and monoids now generalise as follows.  A
  category is {\em M\"obius} if every arrow admits only finitely many
  non-trivial decompositions (of any length).  For a M\"obius category
  $\CC$, the set of arrows $\CC_1$ form a linear basis of its {\em incidence coalgebra},
  where the
  comultiplication of an arrow is the set of all its (length-$2$)
  factorisations
  $$
  \Delta(f) := \sum_{b\circ a = f} a \tensor b  ,
  $$
  immediately generalising the comultiplication of intervals of a 
  poset.  The counit $\epsilon: \CC_1 \to \ground$ sends identity 
  arrows to $1$ and all other arrows to $0$ (again exactly as the 
  case of intervals in a poset).

  The {\em incidence algebra} is the convolution algebra of this 
  coalgebra.
  In here, 
  the {\em zeta function} is the function $\CC_1 \to 
  \ground$ sending every arrow to $1$.  Note that $\epsilon$ is the 
  neutral element for convolution.
\end{blanko}

\begin{theorem}
  (Content--Lemay--Leroux~\cite{Content-Lemay-Leroux})
  For $\CC$ a M\"obius category, the zeta function is convolution 
  invertible with inverse
  $$
  \mu = \Phi_{\operatorname{even}} - \Phi_{\operatorname{odd}} .
  $$
  Here $\Phi_{\operatorname{even}}(f)$ is the number of even-length 
  chains of arrows composing to $f$ (not allowing identity arrows), 
  and similarly for $\Phi_{\operatorname{odd}}$.
\end{theorem}

This alternating-sum formula goes back to Hall~\cite{Hall:1936},
and was also exploited by Cartier and 
Foata~\cite{Cartier-Foata}.\footnote{It is important also
because of its relation to Euler characteristic.
For example, for a finite poset $P$ with a minimal and a maximal 
element added, the alternating-sum formula for the M\"obius function
coincides with the usual formula for Euler characteristic of the
order complex of $P$ (see Stanley~\cite{Stanley:volI}).}
We shall give a slick proof of it in
the abstract setting of the next 
section, where we shall also relate it to
Rota's recursive formula
\begin{equation}\label{recagain}
\mu = \epsilon - \mu * (\zeta-\epsilon)  ,
\end{equation}
valid in any M\"obius category.

\begin{blanko}{Example.}
  In the incidence algebra of the monoid $(\N,+,0)$,
  the M\"obius function is
  \begin{equation}\label{MN}
  \mu(n) = \begin{cases}
    1 & \qquad \text{if } n=0 \\
    -1 & \qquad \text{if } n=1 \\
    0 & \qquad\text{if } n>1 .
  \end{cases}
  \end{equation}
  This is easily proved by checking that this function 
  satisfies the recursion of the general formula~\eqref{recagain}.
  
Hence the inversion principle says in this case
  \begin{eqnarray*}
    \text{if} \qquad f(n) & = & \sum_{k\leq n} g(k)\\
    \text{then} \quad g(n) & = & f(n)-f(n-1) .
  \end{eqnarray*}
  In other words, convolution with the M\"obius function is Newton's
  (backward) finite-difference operator.  So convolution with $\mu$ acts
  as `differentiation' while convolution with $\zeta$
  acts as `integration'.
  If we interpret the sequences $a:\N\to\ground$ as formal power series,
  then the zeta function  is the geometric series, while the M\"obius
  function is $1-x$, as follows from \eqref{MN}.
\end{blanko}

\begin{blanko}{Example.}
  For the monoid $\N^\times$, the incidence algebra is the classical
  algebra of arithmetic functions under Dirichlet convolution,
  recovering the classical M\"obius inversion principle as in 
  Section~\ref{sec:clas}.  Again, the
  closed formula~\eqref{muclas} for the M\"obius function
  can be established easily by simply showing that it
  satisfies the general recursive formula.  A better
  proof explores the fact that the incidence algebra of a product (of M\"obius
  categories) is the tensor product of the incidence algebras, and
  that the M\"obius function of a product is the tensor product of
  M\"obius functions.  Now it follows from unique factorisation of
  primes that $\N^\times$ is the (weak\footnote{Weak means that only
  finitely many factors are allowed to be non-trivial.}) product
  $$
  \N^\times \simeq \prod_p (\N,+)
  $$
  identifying a number
  $
  n = \prod_{p} p^{r_p}  \in \N^\times
  $
  with the infinite vector $(r_2,r_3,r_5,\ldots) \in \prod_p \N$.
  The classical formula~\eqref{muclas} for the M\"obius function now 
  follows as the product 
  of infinitely many copies of 
  the M\"obius function in \eqref{MN}.\footnote{This fact also gives a nice 
  proof of the Euler product expansion 
  (see~\cite[Thm.~280]{Hardy-Wright})
$$
\zeta(s) = \prod_{p \text{ prime}} \frac1 {1-\brok{1}{p^s}}
$$
from 1737 (almost a hundred years earlier than Dirichlet and 
M\"obius).}
\end{blanko}

\begin{blanko}{Example: powersets --- the inclusion-exclusion principle.}
  Let $X$ be a fixed finite set, and consider the powerset of $X$, i.e.~the set
  $\PP(X)$ of all subsets of $X$.  It is a poset under the inclusion relation
  $\subset$.  An interval in $\PP(X)$ is given by a pair of nested subsets of 
  $X$, say $T\subset S$. 
  If the cardinality of $X$ is $n$, then clearly $\PP(X)$ is
  isomorphic as a poset to $\fat{2}^n$, where $\fat 2$ denotes the 
  $2$-element poset
  $[0,1] \subset (\N,\leq)$, so it follows from
  \eqref{MN} and the product rule that
  the M\"obius function on $\PP(X)$ is given by
  $$
  \mu(T\subset S) = (-1)^{\norm{S\!-\!T}} .
  $$
  
  This is the inclusion-exclusion  principle.  As an example of this, consider
  the problem of counting {\em derangements}, i.e.~permutations without 
  fixpoints.  Since every permutation of a set $S$ determines a subset $T$ of points
  which are actually moved, we can write
  $$
  \operatorname{perm}(S) = \sum_{T\subset S} \operatorname{der}(T) 
  $$
  (with the evident notation).
  Hence by M\"obius inversion, we find the formula for derangements
  $$
  \operatorname{der}(S) = \sum_{T\subset S} (-1)^{\norm{S\!-\!T}}
  \operatorname{perm}(T) ,
  $$
  which is a typical inclusion-exclusion formula.
\end{blanko}

\section{Abstract M\"obius inversion}
\label{sec:Mabs}

For background on coalgebras, bialgebras and Hopf algebras, a standard 
reference is 
Sweedler~\cite{Sweedler}.  The little background needed here is amply 
covered also in \cite{Manchon:0408}.

\begin{blanko}{Coalgebras.}
  Let $(C,\Delta,\epsilon)$ be a filtered coalgebra.
Recall that a {\em filtration} of a coalgebra is an increasing sequence
of sub-coalgebras
$$
C_0 \subset C_1 \subset C_2 \subset \cdots  = C
$$
such that
$$
\Delta(C_n) \subset \sum_{p+q=n} C_p \tensor C_q  ,
$$
and recall that an element $x\in C$ is {\em group-like} when 
$\Delta(x) = x\tensor x$;  this implies $\epsilon(x)=1$.
It follows that group-like elements are always of filtration degree 
zero.  We make the following standing assumption 
(see~\cite{Kock:1411.3098}):
\begin{equation}\label{C0-cond}
  \text{\em We assume that $C_0$ is spanned by group-like elements.}
  \end{equation}
\end{blanko}

\begin{blanko}{Convolution algebras.}
  If $(C,\Delta,\epsilon)$ is a coalgebra and $(A,m,u)$ is an algebra, then 
  the space of linear maps $\Lin(C,A)$ becomes
  an algebra under the convolution product: for
  $\alpha, \beta \in \Lin(C,A)$, define $\alpha * \beta$ to be the composite
  $$
  C \stackrel {\Delta} \longrightarrow C \tensor C  \stackrel{ 
  \alpha \tensor \beta} \longrightarrow A \tensor A \stackrel {m} 
  \longrightarrow A,
  $$
  that is, in Sweedler notation~\cite{Sweedler}:
  $$
  (\alpha * \beta) ( x) = \sum_{(x)} \alpha(x_{(1)})\beta(x_{(2)}).
  $$
  The unit for the convolution product is
  $$
  e:=u \circ \epsilon .
  $$
\end{blanko}

\begin{theorem}\label{thm:psiC}\hspace*{-4pt}\footnote{I do not know of any 
  reference for this result.  It may be new, but is in any case a 
  straightforward abstraction of the theorems of Rota and 
  Content--Lemay--Leroux already quoted, once the degree-zero
  condition~\eqref{C0-cond} has been identified~\cite{Kock:1411.3098}.}
  If $\phi\in  \Lin(C,A)$ sends all group-like elements to $1$, then
  $\phi$ is convolution invertible.  The inverse $\psi$ is given by the 
  recursive formula
  \begin{equation}\label{eq:recurs}
  \psi \ = \ e \;-\; \psi * (\phi\!-\!e) .
  \end{equation}
\end{theorem}

We shall give a slick proof consisting mostly of
definitions.\footnote{The proof ingredients go a long way back.  The
even-odd splitting was first used by Hall~\cite{Hall:1936} for
complete lattices, then by Cartier--Foata~\cite{Cartier-Foata} for
monoids, and finally by
Content--Lemay--Leroux~\cite{Content-Lemay-Leroux} for M\"obius
categories, and further exploited in \cite{Lawvere-Menni} and
\cite{Galvez-Kock-Tonks:1512.07577}.  The recursive formula goes back
to Weisner~\cite{Weisner:1935}.  The combined proof is inspired by
\cite{Carlier-Kock}.}

\begin{blanko}{Main Construction.}\label{psi}
  Put $\psy_0 := e$ and
  $$
  \psy_{n+1} := \psy_n * (\phi\!-\!e).
  $$ 
  Put also 
  $$\psy_{\operatorname{even}} := \sum_{n \text{ even}} \psy_n
  \qquad \text{ and }\qquad  \psy_{\operatorname{odd}} := \sum_{n
  \text{ odd}} \psy_n .
  $$
  Finally put 
  $$\boxed{
  \psi :=  \psy_{\operatorname{even}} - \psy_{\operatorname{odd}}.
  }
  $$
  In other words, $\psi = \sum_{n\geq 0} (-1)^n (\phi\!-\!e)^{*n}$.
  This is an infinite sum of functions, but it is nevertheless well 
  defined, because for every input, only finitely many terms in the 
  sum are non-zero.  Indeed, given an element $x\in C$ of filtration 
  degree $r$, then for  $n>r$ the $n$-fold convolution power of the
  $(\phi\!-\!e)$ involves the $n$-fold comultiplication of $x$,
  and since $n>r$ at least one of these factors must be of degree $0$, and 
  hence is killed by $(\phi\!-\!e)$, thanks to the standing 
  assumption~\eqref{C0-cond}.
  
  Now from $\psy_{n+1} = \psy_n * (\phi\!-\!e)$
  we get
  $$
  \psy_{\operatorname{odd}} = \psy_{\operatorname{even}} * (\phi\!-\!e)
  \quad \text{ and }\quad
  \psy_{\operatorname{even}} = e + \psy_{\operatorname{odd}} * 
  (\phi\!-\!e) ,
  $$
  and subtracting these two equations we arrive finally at 
  the formula
  $$\boxed{
  \psi \ = \ e \;-\; \psi * (\phi\!-\!e)
  }$$
  of the theorem.
\end{blanko}

\begin{proof}[Proof of Theorem~\ref{thm:psiC}.]
  First of all, the recursive formula is meaningful: since $\phi$ 
  sends group-like elements to $1$, it agrees with $e$ in filtration
  degree $0$ (thanks to the standing 
  assumption~\eqref{C0-cond}).  Therefore, in the convolution product on the right-hand
  side, $\psi$ is only evaluated on elements of filtration degree 
  strictly less than the element given on the left-hand side.
  Rearranging terms gives $\psi*\phi = e$, showing that $\psi$ is an inverse 
  on the left.
  
  All the arguments can be repeated with $\psy_{n+1} = (\phi\!-\!e) * 
  \psy_n$ (instead of $\psy_{n+1} = \psy_n * (\phi\!-\!e)$), arriving at the right-sided formula $\psi \; = \; e \,-\, 
   (\phi\!-\!e)* \psi$, and rearrangement of the terms shows now that $\psi$ is also an inverse on the right.
\end{proof}

\section{M\"obius inversion in bialgebras}

\label{sec:bialg}

Suppose now that $B$ is a bialgebra, still assumed to be filtered,
and still assumed to have $B_0$ spanned by group-like elements.
Recall that a bialgebra is simultaneously a coalgebra and an algebra,
enjoying in particular the compatibility
\begin{equation}\label{Deltaxy}
  \Delta(x\cdot y) = \Delta(x) \cdot \Delta(y) \qquad \forall x,y .
\end{equation}

For the target algebra $A$, we must now assume it is commutative.
This is used in the proof of Lemma~\ref{mult} below, and 
the rest of the paper depends on that lemma.

\begin{blanko}{Multiplicativity.}
  Call a linear map $\phi: B\to A$ {\em multiplicative}\footnote{Note: in number theory, for arithmetic functions
  $\alpha:\N^\times\to\C$, the word `multiplicative' is used for
  something else, namely the condition $\alpha(mn) = \alpha(m)
  \alpha(n)$ for all $m$ and $n$ relatively prime.  The notions
  are not directly related, because $\N^\times$ is not a bialgebra
  for the usual multiplication:
  for example, $\Delta(2\cdot 2)$ has three terms whereas 
  $\Delta(2)\Delta(2)$ has four terms.
  } if it preserves 
  multiplication:
  $$
  \phi(x\cdot y) = \phi(x) \cdot \phi(y)  \qquad \forall x,y .
  $$
\end{blanko}

\begin{lemma}\label{mult}
  The convolution of two multiplicative functions is again a 
  multiplicative function.  In particular, multiplicative functions 
  form a monoid.
\end{lemma}
\begin{proof}
  This follows immediately from the bialgebra axiom~\eqref{Deltaxy}:
  by expansion in Sweedler notation we have on one hand
  \begin{align*}
  (\alpha *\beta)(xy) &\stackrel{\eqref{Deltaxy}}= \sum_{(x),(y)} \alpha(x_{(1)}y_{(1)})
  \beta(x_{(2)}y_{(2)}) 
  \\
  &= \sum_{(x),(y)} \alpha(x_{(1)}) 
  \alpha(y_{(1)}) \beta(x_{(2)}) 
  \beta(y_{(2)}),
  \end{align*}
  (assuming that $\alpha$ and $\beta$ are multiplicative),
  and on the other hand
  $$
  (\alpha *\beta)(x) (\alpha *\beta)(y) = \sum_{(x),(y)}
  \alpha(x_{(1)}) \beta(x_{(2)}) 
  \alpha(y_{(1)}) 
  \beta(y_{(2)}) .
  $$
  Since $A$ is assumed commutative, these two expressions are equal.
\end{proof}

  Note that multiplicative functions do not form a linear subspace, 
  as the sum of two multiplicative functions is rarely multiplicative.

\begin{lemma}\label{threeterms}
  As before, assume $\phi:B\to A$ sends group-like elements to $1$ and is 
  multiplicative.  Then for any 
  $\alpha:B\to A$ 
  multiplicative we have
  $$
  (\alpha * \phie) (xy) = (\alpha * \phie)(x) \alpha(y) + 
  \alpha(x) (\alpha * \phie)(y) 
  + (\alpha * \phie)(x) (\alpha * 
  \phie)(y) ,
  $$
  where for short we use the temporary notation $\phie := (\phi-e)$.
\end{lemma}
\begin{proof}
  This is simply linearity: substitute $\phi-e$ for $\phie$, expand both sides 
  of the equation, and use multiplicativity of $\alpha$,  $\phi$, and 
  $\alpha*\phi$ (thanks to Lemma~\ref{mult}).
\end{proof}

\begin{prop}\label{psimult}
  Suppose $\phi$ sends group-like elements to $1$, and let $\psi$ 
  denote its convolution inverse.  If $\phi$ is multiplicative,
  then so is $\psi$.
\end{prop}

\begin{proof}
  The proof goes by induction on 
  the degree of $xy$.  If both $x$ and $y$ are 
  group-like (i.e.~degree $0$), it is clear that 
  $\psi(xy)=1=\psi(x)\psi(y)$.
  Now for the induction step.  We use the shorthand notation $\phie 
  := \phi-e$.  First use the recursive 
  formula~\eqref{eq:recurs}:
  $$
  \psi(xy) = - (\psi * \phie)(xy) .
  $$
  Now by induction, the $\psi$ inside the convolution is 
  multiplicative, because its arguments are all  of 
  lower degree (the only case of equal degree in the left-hand 
  tensor factor corresponds to degree $0$ in the right-hand tensor 
  factor, which is
  killed by $\phie=\phi\!-\!e$), so 
  we can apply Lemma~\ref{threeterms}:
  $$= 
  - (\psi * \phie)(x) \psi(y) - 
  \psi(x) (\psi * \phie)(y) 
  - (\psi * \phie)(x) (\psi * 
  \phie)(y) ,
  $$
  and then the recursive equation~\eqref{eq:recurs} backwards (four times):
  $$
  = \psi(x)\psi(y) +
  \psi(x)\psi(y) -
  \psi(x)\psi(y) =  \psi(x)\psi(y).
  $$ 
\end{proof}

\begin{blanko}{Antipodes for bialgebras.}
  Recall that a filtered bialgebra $B$ is {\em connected} if $B_0$
  is spanned by the unit, and that any  connected bialgebras is 
  Hopf~\cite{Sweedler}.
  We shall call $B$ {\em not-quite-connected}~\cite{Kock:1411.3098}
  in the situation where $B_0$ is spanned by group-like 
  elements.
  A notion of antipode for not-quite-connected bialgebras was 
  introduced recently by Carlier and Kock~\cite{Carlier-Kock}.
  It specialises to the usual antipode in the case of a connected 
  bialgebra, and in any case it still serves to compute the M\"obius function as
  $\mu = \zeta \circ S$ as for Hopf algebras.  
  
  In fact, the antipode $S$ itself is an example of 
  abstract M\"obius inversion, as we now proceed to explain.
The idea is simply that one can use the bialgebra $B$ itself as
algebra of values, and invoke abstract M\"obius inversion in
$\Lin(B,B)$.  The identity $B\to B$ does not in general admit a convolution
inverse, because it does not send all group-like elements to $1$.  But
if we just fix that artificially then we can give it as input to the
general construction, and the outcome will be the antipode $S$ in the
sense of \cite{Carlier-Kock}.

To this end, we need to choose $B_+$, a linear complement to $B_0 \subset B$,
and we need to choose it inside $\Ker \epsilon$.
(Note that if the filtration is actually a grading, then 
$B_+$ is canonical, namely the span of all homogeneous elements of 
positive degree.  In practice, $B$ is often of combinatorial nature and
a basis is already given.)
Define the linear operator $T:B\to B$ by
$$
T(x) = 
\begin{cases}
  1 & \text{ if \(x\) group-like } \\
  x & \text{ if \(x \in B_+\)}  .
  \end{cases}
$$
Now apply the Main Construction~\ref{psi}, writing $S$ instead of $\psi$:
$$
S_0 := e, \qquad 
S_{n+1} = S_n * (T-e), \qquad S 
:= S_{\operatorname{even}}-S_{\operatorname{odd}} ,
$$
arriving at the recursion
$$
S \ =\ e \;-\; S * (T\!-\!e) .
$$
By the general M\"obius inversion Theorem~\ref{thm:psiC},
$S$ is the convolution 
inverse to $T$.
But the great feature of this $S$ is that it can invert `anything', by
precomposition.  Precisely:
\end{blanko}

\begin{prop}
  Suppose $\phi: B \to A$ takes group-like elements to $1$, and let 
  $\psi$ denote its convolution inverse. 
  If $\phi$ is multiplicative, then 
  $$
  \psi = \phi \circ S .
  $$
\end{prop}
\begin{proof}
  We calculate
  $$
  \phi *_A (\phi\circ S)  \stackrel{(1)}=
  (\phi \circ T) *_A (\phi\circ S)  \stackrel{(2)}=
  \phi \circ (T *_B S)  \stackrel{(3)}=
  \phi \circ (\eta_B \circ \epsilon)  \stackrel{(4)}=
  \eta_A \circ \epsilon   = e  .
  $$
  Here (1) holds because $\phi$ takes all group-like elements to $1$,
  and $T$ only replaces general group-like elements by 
  the particular group-like element $1$.  Step (2) follows immediately
  from the assumption that $\phi$ is multiplicative. Step (3) is the 
  fact
  that $S$ is convolution inverse to $T$, and step (4) is the fact
  that $\phi$ is unital.
\end{proof}

It is obviously an important property that for multiplicative functions,
M\"obius inversion can be given uniformly by precomposition with the
antipode.  
For this reason,
algebraic combinatorics gradually shifted emphasis from M\"obius inversion
to antipodes~\cite{Schmitt:antipodes}---when they
are available. 
However, we shall see that it is M\"obius inversion that 
generalises to renormalisation, not the antipode.

\normalsize

\section{Direct-sum decomposition and renormalisation}

\label{sec:RB}

Coming back to the case of a coalgebra $C$,
the M\"obius inversion principle says that for every linear function
$\phi:C\to A$ (taking value $1$ on the group-like elements) there
exists another linear function $\psi:C\to A$ that convolves it to the neutral
$e$.

Sometimes one is interesting in less drastic transformations.  For
example, given a linear subspace $K \subset A$, is it possible to
convolve $\phi$ into $K$?  This question is precisely what BPHZ
renormalisation answers: in this case, $A$ is an algebra of
`amplitudes', $K$ is a subalgebra of `finite amplitudes', and  the
result of convolving a map $\phi: C \to A$ into $K$ is {\em
renormalisation}.  
In detail the set-up is the following.

\begin{blanko}{A decomposition problem.}
  Suppose we have a decomposition
$$
A = A_+ \oplus A_-
$$
of $A$ into a direct sum of vector spaces.
Let $R: A \to A$
denote projection\footnote{In \ref{bialgdecomp} below we shall impose the 
Rota--Baxter axiom.} onto $A_-$ relatively to this direct-sum 
decomposition, so that $A_+ = \Ker R$.

Given $\phi : C \to A$ (sending group-like elements to $1$) find
another $\psi: C \to A$ such that $\psi * \phi$ takes values in to $A_+$,
or at least maps $\Ker \epsilon$ to $A_+$.
In other words, find $\psi$ such that $R(\psi*\phi)(x) = 0$, for all 
$x\in \Ker \epsilon$.

This problem can be approached precisely as in the M\"obius inversion 
case (which is the case where $A_-=A$ and $R$ is the identity map).
The only change required is to define a modified  convolution product 
$*_R$ on $\Lin(C,A)$, defined as\footnote{This modified convolution 
product should not be confused with the so-called {\em double product} in the
non-commutative algebra $\Lin(C,A)$, defined as $\alpha 
*_R \beta := R(\alpha)*\beta+\alpha *R(\beta)+\alpha * \beta$.
The double product (in the case where $R$ is Rota--Baxter) plays a role in Lie-theoretic aspects of 
renormalisation~\cite{EbrahimiFard-Guo-Manchon:0602004}.}
$$
\alpha *_R \beta := R(\alpha*\beta) .
$$
Note that $*_R$ is generally neither associative nor 
unital, but  none of these two properties are 
needed in the following main construction.
\end{blanko}

\begin{blanko}{Main Construction.}
  Put $\psy_0 := e$ and 
  $$
\psy_{n+1} := \psy_n *_R (\phi\!-\!e) .
$$
(Note that since $*_R$ is not associative, this is {\em not} the same as 
$(\phi\!-\!e) *_R \psy_n$.  It is important here that all the parentheses
are pushed left.)
As in the classical case, put
  $$\psy_{\operatorname{even}} := \sum_{n \text{ even}} \psy_n
\qquad \text{ and }\qquad  \psy_{\operatorname{odd}} := \sum_{n 
\text{ odd}} \psy_n , 
$$
   and finally 
   \begin{equation}\label{phievenodd}
\boxed{\psi :=  \psy_{\operatorname{even}} -
  \psy_{\operatorname{odd}}}
  \end{equation}
  Just as in the classical case, these are locally finite sums.  This is a 
  consequence of the filtration of $C$---the argument is not affected 
  by the fact that the convolution has been modified.
  
  Now from $\psy_{n+1} = \psy_n *_R (\phi\!-\!e)$
  we get
  $$
  \psy_{\operatorname{odd}} = \psy_{\operatorname{even}} *_R (\phi\!-\!e)
  \quad \text{ and }\quad
  \psy_{\operatorname{even}} = \psy_0 + \psy_{\operatorname{odd}} *_R 
  (\phi\!-\!e) ,
  $$
  and subtracting these two equations we arrive finally at
  the formula
  \begin{equation}\label{REC}
	\boxed{
  \psi \ = \ e \;-\; \psi *_R (\phi\!-\!e)
  }
  \end{equation}
\end{blanko}
  
\begin{lemma}\label{psigr}
  $\psi$ sends group-like elements to $1$.
\end{lemma}
\begin{proof}
  This is clear from \eqref{REC} since $\phi-e$ kills group-like elements.
\end{proof}

\begin{lemma}
  For all $x\in \Ker\epsilon$, we have 
  \begin{enumerate}
    \item[(i)] $\psi(x) \in \Im R =A_-$
  
    \item[(ii)] $(\psi*\phi)(x) \in \Ker R=A_+$.
  \end{enumerate}
  If we assume $1\in A_+$, then (ii) holds for all $x\in C$.
\end{lemma}
\begin{proof}
  Assuming $x\in \Ker\epsilon$,
  the first statement is obvious from \eqref{REC}.
  It follows that we have $R(\psi(x)) = \psi(x)$.
  Rearranging the terms of the recursive equation~\eqref{REC}, we see that
  $$
  \psi *_R \phi = \psi + e - \psi *_R e = \psi + e - R(\psi) .
  $$
  For $x\in \Ker\epsilon$, the right-hand side is zero, whence
  the second statement.  If not $x\in \Ker\epsilon$ then we can 
  assume $x$ group-like, and then $(\psi*\phi)(x) =1$ by
  Lemma~\ref{psigr}.  So then (ii) follows from the alternative 
  assumption $1\in A_+$.
\end{proof}
   
\begin{blanko}{Bialgebra case, Rota--Baxter equation, and multiplicativity.}\label{bialgdecomp}
  For a bialgebra instead of coalgebra, it is natural to demand that
  $\psi$ be multiplicative, provided $\phi$ is so.  To achieve this,
  it turns out one should just demand the direct-sum decomposition $A=
  A_+ \oplus A_-$ to be multiplicative, in the sense that both $A_+$
  and $A_-$ are subalgebras (although not unital subalgebras).
\end{blanko}

\begin{lemma}[Atkinson~\cite{Atkinson}]
  To give such a subalgebra decomposition $A=A_+\oplus A_-$ is
  equivalent to giving an idempotent linear operator $R: A \to A$
  satisfying the 
  {\em Rota--Baxter equation}:\footnote{The equation is more generally written $\theta R(xy) + R(x) R(y)
  = R\big( R(x) y + x R(y)$ for a fixed scalar  weight $\theta$, in 
  order to accommodate the 
  $\theta\!=\!0$ case, which is the equation satisfied by integration by 
  parts.  The equation relevant presently is thus the weight-$1$ 
  Rota--Baxter equation, according to the classical convention.  More 
  recent sources (including~\cite{EbrahimiFard-Guo}, 
  \cite{EbrahimiFard-Manchon}, \cite{EbrahimiFard-Manchon-Patras})  
  tend to use the opposite convention, where the 
  $\theta R(xy)$-term is on the other side of the equation, and the
  weight relevant to BPHZ recursion is thus instead called weight $-1$.}
\begin{equation}\label{RB}
R(x\cdot y) + R(x) \!\cdot\! R(y)= R\big( R(x) \!\cdot\! y + x 
\!\cdot\! R(y) 
\big)  \qquad \forall x,y  .
\end{equation}
\end{lemma}
\begin{proof}
  This check is direct: given $R$, it follows directly from the Rota--Baxter 
equation that both $A_+ := \Ker(R)$ and $A_- := \Im(R)$ are
subalgebras (closed under multiplication).  Conversely, given a
subalgebra decomposition $A=A_+\oplus A_-$, it is easy to check
the Rota--Baxter equation.
\end{proof}

\begin{prop}
  Suppose $\phi$ sends group-like elements to $1$.
  If $\phi$ is multiplicative, then so is $\psi$.
\end{prop}

\begin{proof}
  We use the shorthand notation $\phie:=\phi-e$.  We calculate on one hand
  \begin{eqnarray*}
  \psi(xy) &=& -R\big[ (\psi*\phie)(xy) \big] \\
  &=&
    - R\big[(\psi * \phie)(x) \psi(y) +
  \psi(x) (\psi * \phie)(y) 
  + (\psi * \phie)(x) (\psi * 
  \phie)(y) \big],
  \end{eqnarray*}
  by induction on $\deg(xy)$, using Lemma~\ref{threeterms}, exactly as
  in the proof of Proposition~\ref{psimult} in the classical M\"obius case.
  
  On the other hand, we compute (using \eqref{REC} twice):
  $$
  \psi(x)\psi(y) = \left( -R\big[ (\psi*\phie)(x) \big] \right)
  \left( -R\big[ (\psi*\phie)(y) \big] \right)   .
  $$
  Now apply the Rota--Baxter identity~\eqref{RB}:
  $$
  = -R \Big[
  (\psi*\phie)(x) \cdot (\psi*\phie)(y)
  -   R[(\psi*\phie)(x)] \cdot (\psi*\phie)(y)
-   (\psi*\phie)(x) \cdot R[(\psi*\phie)(y)]
  \Big]
  $$
  and use \eqref{REC} backwards twice:
  $$
  = -R \left[
  (\psi*\phie)(x) \!\cdot\! (\psi*\phie)(y)
  \;+\;   \psi(x) \!\cdot\! (\psi*\phie)(y)
  \;+\;   (\psi*\phie)(x) \!\cdot\! \psi(y)]
  \right] ,
  $$
  which agrees with the computation of $\psi(xy)$.
\end{proof}

\begin{blanko}{Non-concluding historical remarks.}
  Equation~\eqref{REC} is the abstract BPHZ recursion of
  Section~\ref{sec:BPHZ}, often called the Bogoliubov recursion.  The
  Hall--Leroux style even-odd formula
  $$
  \psi = \psi_{\operatorname{even}}-\psi_{\operatorname{odd}}
  $$
  of Equation~\eqref{phievenodd} features less prominently in 
  renormalisation theory, see \cite{EbrahimiFard-Kreimer:0510} and
  \cite{EbrahimiFard-Manchon}.
  Expanded, it says
  $$
  \psi = \sum_{n\geq 0} (-1)^n \psy_n = \sum_{n\geq 0}(-1)^n R(R(R( \dots *
  \phie )*\phie)*\phie) ,
  $$
  where the $n$th term of the sum has $n$ applications of $R$ and $n$ 
  convolution factors, and
  where as usual we use the shorthand $\phie := \phi-e$.  This is the
  solution of Atkinson~\cite{Atkinson} to the factorisation problem
  posed by $R$.  Atkinson actually uses the abstract form of the
  `Bogoliubov' recursion in his Second Proof~\cite{Atkinson}, in a way
  similar to the proofs above.  The equivalence between Atkinson's
  formula and the Bogoliubov recursion has been exploited further in
  the context of renormalisation and Lie theory by Ebrahimi-Fard,
  Manchon and Patras~\cite{EbrahimiFard-Manchon-Patras}.  It is 
  striking that it comes about from the two aspects of general M\"obius 
  inversion.

\begin{center}
  * * *
\end{center}

  Frederick Atkinson spent the first part of his mathematical life working
  in analytic number theory, contributing in particular to the theory of
  arithmetic functions and Dirichlet series.
  His 1949 paper with Cherwell~\cite{Atkinson-Cherwell} (cited in Hardy 
  and Wright~\cite{Hardy-Wright}), is about average values of arithmetic 
  functions related by M\"obius inversion.
  In the 1950s his interests shifted to functional analysis 
  and operator theory, which was the context for his interest in Baxter's
  work, leading to his 1963 paper~\cite{Atkinson} already
  mentioned.  For more information about Atkinson's life and work, see~\cite{Mingarelli}.

  \medskip
  
  The notion of Rota--Baxter algebra had been introduced by Glen
  Baxter~\cite{Baxter} in fluctuation theory of sums of random variables in
  1960.  Rota, Cartier, Foata, and others realised the usefulness of the
  notion (at the time called {\em Baxter algebras}\footnote{The 
  renaming to {\em Rota--Baxter algebra} occurred in 
  \cite{EbrahimiFard:0207043}, marking also the first connection 
  between this subject and renormalisation.}) also in algebra and
  combinatorics, notably in the theory of symmetric functions, and
  Rota~\cite{Rota:BaxterI} used elementary category theory to unify several
  results by establishing them in the free Rota--Baxter algebra.  For a
  glimpse into the extensive theory of Rota--Baxter algebras, with emphasis
  on their use in renormalisation, see \cite{EbrahimiFard-Guo,EbrahimiFard-Guo-Kreimer,
  EbrahimiFard-Guo-Manchon:0602004}.
 
  \medskip
  
  Gian-Carlo Rota was a main character both in the development of M\"obius
  inversion and in the development of Rota--Baxter algebras, in both
  cases making these constructions into general tools.  Naturally, he
  also combined these two toolboxes: for example, in his 1969 proof of the
  so-called Bohnenblust--Spitzer identity (see 
  also~\cite{EbrahimiFard-Manchon-Patras}) in the
  free (and hence in every) Rota--Baxter 
  algebra~\cite{Rota:BaxterII}, a
  key point is showing that the signs in that formula arise from the
  M\"obius function of the partition lattice~\cite{Rota:Moebius}.
  Rota did not have the idea of
  entangling the Rota--Baxter operator with the recursions of M\"obius
  inversion itself, though.  From the ahistorical viewpoint of the present
  contribution, this is what Bogoliubov~\cite{Bogoliubov-Parasiuk} and
  Atkinson~\cite{Atkinson} achieved---without having the general
  theory of M\"obius inversion at their disposal.
\end{blanko}

\bigskip
  
  \footnotesize
  
  \noindent {\bf Acknowledgments.} I wish to thank Kurusch
  Ebrahimi-Fard for helping me with renormalisation and many related topics
  over the past decade, and more specifically for many pertinent
  remarks on this manuscript.  Thanks also to Dominique Manchon and
  Fr\'ed\'eric Patras for helpful feedback.  Support from grants
  MTM2016-80439-P (AEI/FEDER, UE) of Spain and 2017-SGR-1725 of
  Catalonia is gratefully acknowledged.

\hyphenation{mathe-matisk}

\end{document}